\documentclass{fundam}

\publyear{25}
\papernumber{1}
\volume{194}
\issue{4}
\theDOI{10.46298/fi.14339}

\versionForARXIV

\usepackage{mymacros}
\begin{document}
\title{Ciphertext malleability in Lattice-Based KEMs as a countermeasure to Side Channel Analysis}
\author{Pierre-Augustin Berthet\thanks{Supported by Agence de l'Innovation de D\'efense, grant 2022156 Th\`ese CIFRE D\'efense}\corresponding\\
Hensoldt France SAS, Plaisir, France\\
LTCI, T\'el\'ecom Paris, Palaiseau, France\\
berthet{@}telecom-paris.fr
}

\maketitle
\runninghead{P-A. Berthet
}{Ciphertext malleability in Lattice-Based KEMs as a countermeasure to SCA}

\vspace{-7ex}

\begin{abstract}
    Due to developments in quantum computing, classical asymmetric cryptography is at risk of being breached. Consequently, new Post-Quantum Cryptography (PQC) primitives using lattices are studied. Another point of scrutiny is the resistance of these new primitives to Side Channel Analysis (SCA), where an attacker can study physical leakages. In this work, we discuss an SCA vulnerability due to the ciphertext malleability of some PQC primitives exposed in a work by Ravi et al. We propose a novel countermeasure to this vulnerability exploiting the same ciphertext malleability and discuss its practical application to several PQC primitives. We also extend the seminal work of Ravi et al. by detailing their attack on the different security levels of a post-quantum Key Encapsulation Mechanism (KEM), namely FrodoKEM. We also provide a generalisation of their attack to different parameters which could be used in future similar primitives.
\end{abstract}

\section{Introduction}
With the emergence of quantum computing, new primitives for asymmetric cryptography are studied. The National Institute of Standards and Technology (NIST) launched a competition \cite{alagic2022status,chen2016report} to select a future Post-Quantum Cryptography (PQC) standard. The first Key Encapsulation Mechanism (KEM) to be standardised is the Module Lattice-based Key Encapsulation Mechanism (ML-KEM) or FIPS 203 \cite{nist2023mlkem}. This KEM is based on the CRYSTALS-Kyber candidate \cite{avanzi2021crystals}. However, other finalists from the third round of competition are using similar mathematical problems, such as SABER \cite{basso2020saber} and FrodoKEM \cite{naehrig2017frodokem}.


Another concern for new PQC primitives is their resistance to physical attacks. First introduced by Kocher \cite{kocher1996timing}, Side-Channel Analysis (SCA) is a kind of cryptanalysis which targets the physical environment of a cryptosystem and analyses potential leakages (e.g. execution timing, heat, power consumption, electromagnetic emissions) to try and discover correlations between sensitive data and the physical data the attacker observes.


The resistance of PQC primitives to SCA has already been studied. Several works summarise existing papers in terms of attacks and possible countermeasures, such as Ravi et al.\cite{ravi2024side} and Canto et al.\cite{canto2023algorithmic} for ML-KEM. Most noticeably, ciphertext malleability, a vulnerability of several NIST PQC lattice-based candidates, is presented by Ravi et al.\cite{ravi2021exploiting}.

\subsection{Our contributions}
In this work, we present a countermeasure to side-channel opponents targeting the decoding of the message in PQC Lattice-based KEM. The attacker uses ciphertext malleability to do so and, according to Ravi et al. \cite{ravi2021exploiting}, generic countermeasures such as masking \cite{chari1999towards,ishai2003private,prouff2013masking} or shuffling \cite{herbst2006aes,veyrat2012shuffling} are not sufficient. The countermeasure uses the same ciphertext property to perform the attack as a defence mechanism. As a consequence, the attacker cannot exploit ciphertext malleability to retrieve the secret.


However, from a high-level perspective, the countermeasure reuses existing primitive functions. Although side-channel attacks against these functions can circumvent the countermeasure, it is fully compatible with other generic countermeasures. We recommend its deployment alongside a shuffling countermeasure to have a message side-channel resistance similar to a masking and shuffling countermeasure but at lower costs. In essence, the countermeasure presented in this work ``displaces'' the leakage point from a function where generic countermeasures are not effective to one where they are.


The countermeasure initially intervenes in the \Fdec procedure of the KEM. We also discuss extending its application to the remainder of the \KDec procedure. We discuss the specificities of this extension for three NIST PQC competition candidates, namely SABER, FrodoKEM and ML-KEM.


We also extend the seminal attack work from Ravi et al. \cite{ravi2021exploiting}, originally against ML-KEM \cite{nist2023mlkem}, by providing a more detailed insight on how the attack can be performed against the different parameter settings of the post-quantum algorithm FrodoKEM. Ravi et al. \cite{ravi2021exploiting} attack works by adding biases to the ciphertext to flip specific bits of the output of the decoding of the message during decapsulation. Assuming 8-bit (resp. 16-bit) storage of the message in memory, we can recover the message of Frodo-640 in 9 (resp. 17) traces. For Frodo-976, the adaptation of Ravi et al. \cite{ravi2021exploiting} attack requires 10 (resp. 19) traces. We propose more optimal choices of biases than the adaptation of Ravi et al. \cite{ravi2021exploiting} attack by studying the impact of all possible biases instead of only focussing on the ones that flip bits. Our ciphertext choices only require 8 (resp. 14) traces for Frodo-976. For Frodo-1344, the adaptation requires 9 (resp. 17) traces, and our improved choice of biases only 7 (resp. 13).


The optimal choices of biases in our improvement of Ravi et al.\cite{ravi2021exploiting} attack were done heuristically by brute-forcing all possibilities. A Jupyter Notebook with our methods is available at the following link:
\begin{center}
    \url{https://github.com/Pierre-Augustin-Berthet/decode-lpr}
\end{center}
The paper is structured as follows. In Section \ref{sec:prel} we introduce notations and detail the LPR framework and the different primitives that use this framework in the third round of the NIST PQC competition. Ciphertext malleability is presented in Section \ref{sec:attack} with the seminal work of Ravi et al. \cite{ravi2021exploiting} and our further extension of their attack to the FrodoKEM algorithm. A generalisation of the attack is also discussed. The countermeasure is discussed throughout Section \ref{sec:countermeasure}, with descriptions of its overall strategy, extension to the entirety of \KDec, scalability, compatibility with other generic countermeasures, as well as a discussion on its impact against other SCAs. Section \ref{sec:conclusion} concludes the paper.

\section{Preliminaries}\label{sec:prel}
\subsection{Notations}\label{sec:notation}
We denote by $q\in\mathbb{N}$ a modulus and $\mathbb{Z}_q$ the set $\mathbb{Z}/q\mathbb{Z}$. $n\in\mathbb{N}$ a message length, and $R_q$ the devolution ring $\mathbb{Z}_q[X]/<X^n+1>$. We denote by $\mathcal{B}$ the set of eight bits (or byte) $\{0,1\}^{8}$ or $\mathbb{F}_2^8$. We denote by $\mathcal{U}(\mathcal{K})$ the uniform distribution over a set $\mathcal{K}$.

\begin{remark}
    In this paper, we discuss three different primitives, each with their specific and conflicting notation. Their notation will be used for their description; however, we unify the notation for the other sections of the paper.
\end{remark}


The Hamming Weight of a binary message $m$ is the number of non-zero bits it contains. We denote it $HW(m)$. We denote the difference between the Hamming Weight of a message $m$ and of a biased message $m'$ by $\Delta=HW(m')-HW(m)$. 


We introduce the concept of an X-classes distinguisher. Let $\mathcal{I}$ be an input set for a function $F$ and $\mathcal{S}$ the output set. We have a X-classes distinguisher if we can build X distinct classes of subsets of $\mathcal{I}$ from subsets of $\mathcal{S}$ such as $F(\mathcal{I})=\mathcal{S}$. Consequently, a class is denoted by $S\in\mathcal{S}\rightarrow I\in\mathcal{I}$, that is, the observation of $S$ implies that the input of $F$ has the shape $I$. In this paper, $\mathcal{I}$ contains binary words composed of the letters $0$ and $1$. The letter $\centerdot$ is used if the information on a letter of the binary word is not certain. On the other hand, elements of $\mathcal{S}$ are composed of a sign ``$+$'' or ``$-$'' followed by a positive integer. For example, if we consider \begin{equation}
    F:\mathcal{I} \rightarrow \mathcal{S},F(00)=+1,F(01)=+0,F(10)=+1,F(11)=-1,
\end{equation}
we have a two-classes as well as a three-classes distinguishers. The two-classes distinguisher is $\{+1\rightarrow \centerdot0 ,(+0,-1)\rightarrow \centerdot1 \}$. The three-classes one is $\{+1\rightarrow \centerdot0 ,+0\rightarrow 01 ,-1\rightarrow 11 \}$. 

\subsection{Generic SCA countermeasures}
To avoid correlations between side channel leakage and sensitive variables, several generic countermeasures have been proposed. They tend to add randomness in one way or another to the computations.


For example, \emph{shuffling} \cite{herbst2006aes,veyrat2012shuffling} randomises the execution order of the algorithm whenever possible. Although this does not eliminate leakage, leakage detection often requires several attempts and data alignment between those attempts to be successful. Shuffling foils the alignment and thus increases the number of attempts required to successfully exploit the leakage.


The core idea behind \emph{blinding} \cite{mamiya2004efficient,saarinen2018arithmetic} is to add to the computations a random value that will be eliminated in the later stages of the function. As a consequence, it randomises the side-channel leakage of the function for each of its iterations. Blinding methods tend to be specific to the functions they protect and are less generic.


Finally, \emph{masking} \cite{chari1999towards,ishai2003private,prouff2013masking} aims to divide the sensitive data into several random \emph{shares}. Those shares are processed by the same algorithm separately and reassembled in the later stages of the protected function to ensure the correctness. This forces a side-channel attacker to use more probing hardware as it must recover each share to breach the implementation. It can be considered the most studied and popular generic countermeasure as it benefits from formal proofs and models. A simple example of masking is the \emph{boolean masking}, which masks bits:
\begin{equation}
    \text{For }x\in\{0,1\}\text{ and }r\leftarrow\mathcal{U}(\{0,1\}), Mask(x) = (x \oplus r,r), Unmask((x\oplus r,r))=(x\oplus r)\oplus r
\end{equation}

\begin{remark}
    The generic countermeasures presented here all intervene at an algorithmic level. Other generic countermeasures can be set up at lower levels, including directly at a hardware level with \emph{shielding} \cite{plos2008enhancing}, for example. Shielding is performed by putting a physical barrier between the opponent and the targeted hardware, avoiding any leakage.
\end{remark}

\subsection{LWE/LWR PKE using lattices}
Several NIST PQC candidates rely on the Learning With Error (LWE) problem on lattices and its variant the Learning With Rounding (LWR). They are based on the Luybashevsky-Peikert-Regev (LPR) framework \cite{10.1007/978-3-642-13190-5_1}, described in Algorithm \ref{alg:lpr}. The error distribution on a set $\mathcal{K}$ for LPR is indicated by $\chi(\mathcal{K})$. It can be deterministically computed from a seed $r$, in this case it will be denoted by $\chi(r,\mathcal{K})$.

\begin{algorithm}[h!]
    \caption{LPR Encryption Scheme \cite{ravi2021exploiting,10.1007/978-3-642-13190-5_1}}
    \label{alg:lpr}
    \DontPrintSemicolon
    \Fn{\Fkg{}}{
        $\textbf{a}\leftarrow \mathcal{U}(\Z_q)$\;
        $\textbf{s},\textbf{e} \leftarrow \chi(\Z_q)$\;
        $\textbf{t} = \textbf{a} \times \textbf{s} + \textbf{e}$\;
        \KwRet{pk $= (\textbf{a},\textbf{t})$, sk $= (\textbf{s})$}\;
    }
    \hrulefill\;
    \Fn{\Fenc{$pk$,$m\in \mathcal{B}^{32}$,$r\in\mathcal{B}^{32}$}}{
        $\textbf{s}',\textbf{e}',\textbf{e}''\leftarrow\chi(r,\Z_q)$\;
        $\textbf{u} = \textbf{a} \times \textbf{s}' + \textbf{e}'$\;
        $\textbf{v}' = \textbf{t} \times \textbf{s}' + \textbf{e}''$\;
        $\textbf{x} = \Enco(m,d)$\;
        $\textbf{v} = \textbf{v}' + \textbf{x}$\;
        \KwRet{ct $=(\textbf{u},\textbf{v})$}\;
    }
    \hrulefill\;
    \Fn{\Fdec{$ct$,$sk$}}{
        $\textbf{x}' = (\textbf{v} - \textbf{u}\times\textbf{s})$\;
        $m'= \Deco(\textbf{x}',d)$\;
        \KwRet{$m'$}\;
    }
\end{algorithm}

The \Enco function is a compression function. It corresponds for a parameter $d$ to a mapping of elements of $\Z_{2^d}$ to $\Z_q$. The function \Deco performs a mapping of elements of $\Z_q$ to $\Z_{2^d}$. They use the rounding to the nearest integer denoted by $\nint{\cdot}$. When applied to a matrix or a vector or a polynomial, these functions are applied separately on each coefficient of the matrix, vector, or polynomial. This rule applies to matrix and vector of polynomials, these functions are applied separately on the coefficients of the polynomials. They are defined as follows:
\begin{eqnarray}
    \label{eq:decomp}
    \forall\beta\in\Z_{2^d},\ \Enco(\beta,d) =& \!\!\!\nint*{ \frac{q}{2^{d}}\cdot \beta} \ mod\ q, \\
	\label{eq:comp}
	\forall\alpha\in\mathbb{Z}_q,\ \Deco(\alpha,d) =& \nint*{\frac{2^{d}}{q} \cdot \alpha} \ mod\ 2^{d}.
\end{eqnarray}
We can also use a sector representation to illustrate the effect of the \Deco function as seen in Figure \ref{fig:deco}. The possible inputs are represented as positions on the circle boundary, with values ranging from $0$ to $q\equiv0$ counterclockwise. As the output depends on which interval the input is part of, the interval is drawn as a sector of the circle, and the sector is labelled with its corresponding output. For example, in Figure \ref{subfig:quad2}, if the input of \Deco is between $\frac{q}{8}$ and $\frac{3q}{8}$, it will be mapped to $01$.

\begin{figure}[h!]
    \centering
    \begin{subfigure}[c]{0.32\textwidth}
        \centering
        \includegraphics[width=\textwidth]{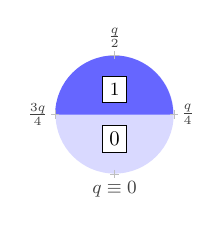}
        \subcaption{$d=1$}
        \label{subfig:quad1}
    \end{subfigure}
    \begin{subfigure}[c]{0.32\textwidth}
        \centering
        \includegraphics[width=\textwidth]{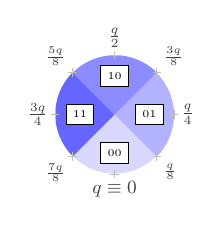}
        \subcaption{$d=2$}
        \label{subfig:quad2}
    \end{subfigure}
    \begin{subfigure}[c]{0.32\textwidth}
        \centering
        \includegraphics[width=\textwidth]{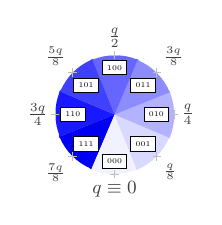}
        \subcaption{$d=3$}
        \label{subfig:quad3}
    \end{subfigure}
    \caption{Sector representation of the decoding function for different parameters $d$, with the intervals on the circle boundary counter-clockwise and the sectors labeled with the corresponding outputs to the intervals}
    \label{fig:deco}
\end{figure}

\subsection{Fujisaki-Okamoto Transform}
The LPR framework is not IND-CCA secure (INDistinguishable-Chosen Ciphertext Attack) as an attacker can modify $\textbf{u}$ and $\textbf{v}$ to recover information on $\textbf{s}$ from $m'$. To counteract this attack, LPR-based primitives use an adaptation of the Fujisaki-Okamoto Transform (FOT) \cite{fujisaki1999secure} as described in Algorithm \ref{alg:fot}.

\medskip

\begin{algorithm}[H]
    \caption{Fujisaki-Okamoto Transform for LPR \cite{ravi2021exploiting,10.1007/978-3-642-13190-5_1}}
    \label{alg:fot}
    \small
    \DontPrintSemicolon
    \Fn{\KKG{}}{
        $(pk,\textbf{s}) =$ \Fkg{}\;
        $z\leftarrow\mathcal{U}(\mathcal{B}^{32})$\;
        \KwRet{pk, sk $=(\textbf{s},z)$}\;
    }
    \hrulefill\;
    \Fn{\KEnc{$pk$}}{
        $m\leftarrow\mathcal{U}(\mathcal{B}^{32})$\;
        $r = \PRF{m,pk}$\tcp*{Pseudorandom Function}
        $ct = \Fenc{pk,m,r}$\;
        $K = \KDF{r,ct}$\tcp*{Key Derivation Function}
        \KwRet{ct,K}\;
    }
    \hrulefill\;
    \Fn{\KDec{$ct$,$pk$,$sk$}}{
        $m' = \Fdec{sk,ct}$\;
        $r' = \PRF{m',pk}$\;
        $ct' = \Fenc{pk,m',r'}$\;
        \If{ct==ct'}{
            \KwRet{K=\KDF{$r',ct'$}}\;
        }
        \Else{
            \KwRet{K=\KDF{$z,ct'$}}\;
        }
    }
\end{algorithm}
\subsection{FrodoKEM}
FrodoKEM \cite{naehrig2017frodokem} is a NIST PQC third-round alternate finalist. Although it was not selected for the fourth round, it is selected by the German Bundesamt für Sicherheit in der Informationtechnick (BSI) and the French Agence Nationale de la Sécurité des Systèmes d'Information (ANSSI) as a more conservative option than ML-KEM. It was noticeably used in the first French diplomatic telegram using post-quantum cryptography on 30 November 2022, and is currently submitted to the International Organisation for Standardisation (ISO). An Internet-Draft (I-D) has also been submitted to the Internet Engineering Task Force (IETF). Thus, it is still relevant to study its resistance to SCAs.

\subsubsection{Parameters}
\begin{remark}
    The notations used in the description of each primitive in this paper are the ones used in their respective specification papers. However, outside of their description, we use the unified notation defined in Section \ref{sec:notation}.
\end{remark}
FrodoKEM relies on the LWE problem and the LPR framework. It uses matrices over $\Z_q^{\bar{m}\times n}$, $\Z_q^{\bar{m}\times \bar{n}}$ and $\Z_q^{\bar{n}\times n}$.
FrodoKEM offers three levels of security and different sets of parameters for each level:
\begin{table}[H]
    \centering
	\caption{Parameters sets for FrodoKEM}\label{tab:frodo_params}
	\begin{tabular}{|l|c|c|c|c|c|}
		\hline
		& NIST security level & $n$ & $q$ & $\bar{m}\times\bar{n}$ & $B$ \\
		\hline
		Frodo-640 & I & 640 & $2^{15}$ & $8\times8$ & 2\\
		Frodo-976 & III & 976 & $2^{16}$ & $8\times8$ & 3\\
		Frodo-1344 & V & 1344 & $2^{16}$ & $8\times8$ & 4\\
		\hline
	\end{tabular}
\end{table}

The most interesting parameter for this work is the compression factor $B$ applied to the message in FrodoKEM. It is equivalent to the parameter $d$ from Equations \ref{eq:decomp} and \ref{eq:comp}. The compression functions are defined as $Frodo.Encode$ for \Enco and $Frodo.Decode$ for \Deco.

\begin{remark}
    In FrodoKEM, the parameter $n$ does not denote the message length but rather a matrix size parameter. The message length is given by $l=B\times\bar{m}\times\bar{n}$.
\end{remark}
\subsection{ML-KEM}
ML-KEM \cite{nist2023mlkem} is a slight modification of CRYSTALS-Kyber \cite{avanzi2021crystals}, a KEM selected by NIST \cite{alagic2022status}. ML-KEM is the first post-quantum KEM standard published by the NIST.

\subsubsection{Parameters}
Its design adapted from LPR relies on several instances of the Module-LWE (M-LWE) and Module-LWR (M-LWR) problems. It uses vectors in $R_q^k$ and square matrices in $R_q^{k\times k}$. ML-KEM offers three levels of security and different sets of parameters for each level:
\begin{table}[H]
    \centering
	\caption{Parameters sets for ML-KEM}\label{tab:kyb_params}
	\begin{tabular}{|l|c|c|c|c|c|c|c|c|c|}
		\hline
		&  NIST security level  & $n$ & $q$ &  $k$  & $\eta_1$ & $\eta_2$  & $d_u$ & $d_v$ & $d_m$ \\
		\hline
		ML-KEM-512 &  I &  256 & 3329 &  2 &  3 & 2  & 10 & 4 & 1\\
		ML-KEM-768 &  III  & 256 & 3329  & 3  & 2 & 2  & 10 & 4 & 1\\
		ML-KEM-1024 & V & 256 & 3329 &  4 & 2 & 2  & 11 & 5 & 1\\
		\hline
	\end{tabular}
\end{table}

The compression factor applied to the message in ML-KEM is denoted $d_m$ in Table \ref{tab:kyb_params} as ML-KEM does not use a specific notation for the parameter $d$ from Equations \ref{eq:comp} and \ref{eq:decomp} when applied to the message, as it is always $1$. Another interesting point is the use of compression to reduce the size of the ciphertext in ML-KEM, the relevant parameters being $d_u$ and $d_v$. The compression functions are defined as $Decompress_q$ for \Enco and $Compress_q$ for \Deco.
\subsection{SABER}
SABER \cite{basso2020saber} is a NIST PQC third round finalist. It was not selected for the fourth round nor as a standard by the NIST or the BSI. Its future deployement for real world application is uncertain. Nonetheless, its similarities with ML-KEM makes its study relevant. Especially, it relies on a power of $2$ modulus rather than a Solinas/Crandall prime modulus like ML-KEM and uses the M-LWR problem rather than M-LWE.

\subsubsection{Parameters}
SABER offers three levels of security and different sets of parameters for each level:
\begin{table}[H]
    \centering
	\caption{Parameters sets for SABER}\label{tab:sab_params}
	\begin{tabular}{|l|c|c|c|c|c|c|c|}
		\hline
		&  NIST security level  & $n$ & $q$ & $p$ & $T$ & $l$  & $\mu$ \\
		\hline
		LightSaber &  I &  256 & $2^{13}$ & $2^{10}$ & $2^3$ & 2 & 10\\
		Saber &  III  & 256 & $2^{13}$  & $2^{10}$ & $2^4$ & 3 & 8\\
		FireSaber & V & 256 & $2^{13}$ & $2^{10}$ & $2^6$ & 4 & 6\\
		\hline
	\end{tabular}
\end{table}

As SABER relies only on M-LWR, it uses several hidden compressions which are defined as functions (\Deco and \Enco) in the specification paper. The compression for the encoding of the message is tied to the parameter $p$. The parameter $T$ is used for ciphertext compression, a method also used in ML-KEM.

\section{Exploiting ciphertext malleability: existing work and generalisation}
\label{sec:attack}
\subsection{Seminal work: Targeted message bit flip for d=1}
Ravi et al.\cite{ravi2021exploiting} introduced an attack exploiting what they called ``ciphertext malleability'' in the LPR framework. They perform a chosen-ciphertext attack with the aim of flipping a bit of the secret message. Given a side-channel Oracle that we will denote by $\mathcal{O}_{SCA}$, they can recover the exact value of any message bit. 


The oracle is defined as follows: When queried, $\mathcal{O}_{SCA}$ provides the Hamming Weight of a register where an output of \Fdec is stored. This can be performed in side-channel against any implementation on a microcontroller which stores the message computed by \Fdec within a group of registers in memory before further computations. Therefore, the attack can be easily performed in a real-world setting. In their work, Ravi et al. \cite{ravi2021exploiting} consider a perfect Hamming Weight distinguisher as $\mathcal{O}_{SCA}$. If such a distinguisher is not available, they recommend using an imperfect one several times and to perform a majority vote.


The main metric used in side-channel attacks is the number of traces, i.e. measurements, required to recover the targeted secret. This number depends on many parameters, e.g. the Signal-to-Noise Ratio (SNR) of the measures or the size of the registers in which the secret is stored. For example, the seminal attack by Ravi et al. \cite{ravi2021exploiting} requires $256$ calls to the oracle against MLKEM \cite{nist2023mlkem}, without the calls to recover the original message Hamming Weight. Taking into account the storage of the message in 8-bit words and an optimal SNR, they parallelise the calls to the oracle and reduce the attack cost to only $9$ traces. In this work, we consider the performance of the attack using both traces and number of queries to the oracle, as some countermeasures like shuffling require to recover the Hamming Weight of the entire message for the attack to still work, and thus remove the possibility of using parallel attacks. The Hamming Weight of the entire message can be recovered by adding the Hamming Weights of its registers.



The ciphertext malleability is due to the function \Deco. When its input is shifted by adding a specific bias modulus $q$, this can alter the output in a predictable way. Using $\mathcal{O}_{SCA}$ before and after the chosen ciphertext attack, it is possible to deduce the exact value of any bit of the message. The biases have the form $+\frac{qk}{2^d}$, where $k<2^d$ an integer and $d$ the compression factor of the function \Deco. In the case of MLKEM \cite{nist2023mlkem}, as $d=1$ and $q$ is not a multiple of $2^d$, the bias injected is rounded, i.e. $+\nint*{\frac{q}{2}}=+\frac{q+1}{2}=1665$.


The attack by Ravi et al. \cite{ravi2021exploiting} against MLKEM is carried out as follows:
\begin{enumerate}
    \item The Hamming Weight of the original message $m$ is recovered with $\mathcal{O}_{SCA}$.
    \item The bias $+\nint*{\frac{q}{2}}$ is injected into the ciphertext to target a specific bit of the message. To target the $i^{th}$ bit of the message, we add the bias to the ciphertext \textbf{v} during the \Fdec procedure at the start of \KDec in order to force the computation of $\Deco(\textbf{v}+\nint*{\frac{q}{2}}X^i-\textbf{u}\times\textbf{s})$.
    \item The target device performs its computations and a call to $\mathcal{O}_{SCA}$ retrieves the Hamming Weight of the biased message $m'$.
    \item The results are compared: $\Delta = HW(m') - HW(m)=\begin{cases}
        +1\text{ if the message bit is }0;\\
        -1\text{ if the message bit is }1.
    \end{cases}$
\end{enumerate}

\begin{remark}
    For a parallelised attack, the oracle $\mathcal{O}_{SCA}$ only needs to recover the Hamming Weight of the register which contains the biased part of the message. Hence, the attack can be performed on each register separately in parallel to lower the number of traces.
\end{remark}

The representation in sectors allows for a better understanding of how ciphertext malleability can be exploited. Adding a specific bias to the input of \Deco ``rotates'' the sectors and alters the output accordingly. An example is given in Figure \ref{fig:compress}, where one can bit-flip the output of the \Deco function with parameter $d=1$ by adding $\frac{q}{2}$ to the input.

\begin{figure}[H]
    \centering
    \begin{subfigure}[c]{0.4\textwidth}
        \centering
        \includegraphics[height=4cm]{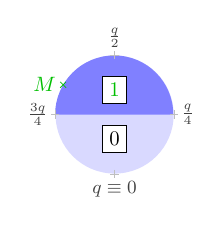}
        \subcaption{$\Deco(M,1)=1$}
        \label{subfig:compressprebf}
    \end{subfigure}
    \begin{subfigure}[c]{0.4\textwidth}
        \centering
        \includegraphics[height=4cm]{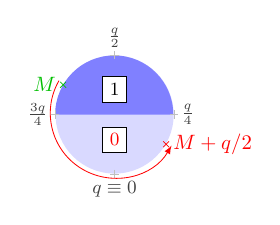}
        \subcaption{$\Deco(M+\frac{q}{2},1)=\bar{1}=0$}
        \label{subfig:compresspostbf}
    \end{subfigure}
    \caption{Bitflipping the output of the decoding for parameter $d=1$}
    \label{fig:compress}
\end{figure}

\subsection{Generalising targeted message bit flip}\label{subsec:gene}
Ravi et al. \cite{ravi2021exploiting} focused on the ML-KEM primitive, where $d=1$. In this section, we extend their targeted bit flip strategy to any value of the parameter $d$. A higher value of $d$ also implies that a single call to \Deco($\alpha,d$) outputs several bits of message at a time. 


First, we observe that, for any integer $k<2^d$, \begin{align*}
    \Deco(\alpha + \frac{qk}{2^d},d) &\equiv \nint*{\frac{2^d}{q}(\alpha+\frac{qk}{2^d})} \ mod\ 2^d\\
    &\equiv \nint*{\frac{2^d}{q}\alpha+k}\ mod\ 2^d\\
    &\equiv  \nint*{\frac{2^d}{q}\alpha}+ k\ mod\ 2^d\\
    &\equiv \Deco(\alpha,d)+k\ mod\ 2^d
\end{align*}

When $d=1$, this results in XORing the single bit of the output of \Deco($\alpha,1$) with $1$. For higher values of $d$, $k$ is added to the output. When $k$ is a power of $2$, we have the following theorem:

\begin{theorem}
    For any value of $d$, there exists a two-classes distinguisher that gives the value of at least one bit of the output of \Deco($\alpha,d$) with absolute certainty, and this for each bit of the output.
\end{theorem}

\begin{proof}
    Let $k=2^j$ be the power of $2$ for an integer $j$. Let $m=\Deco(\alpha,d)$ and $m'=\Deco(\alpha,d)+k\ mod\ 2^d$. Let $\Delta = HW(m')-HW(m)$. We use the following binary description: $m=\sum_{i=0}^{d-1}m_i*2^i$. Adding $k$ to $m$ has the following results:
    \begin{equation}\label{eq:carryprop}\begin{cases}
        \text{if }m_j=0,\ \text{then }m'_j=1\text{ and }\Delta=+1;\\
        \text{if }m_j=1,\ \text{then }m'_j=0\text{ but we have a carry to propagate.}
    \end{cases}\end{equation}
    If there is a carry to propagate, we have $\Delta = -1 + \epsilon$, with $\epsilon$ the impact of carry propagation on the Hamming Weight. The propagation of the carry can be seen as adding $k' =2k$ to a variant of $m$, denoted $m''$, such as $\forall i\neq j,m_i=m''_i$ and $m''_j=0\neq m_j$. Thus, reapplying the results from Equation \ref{eq:carryprop}, we have $\epsilon=+1$ and $\Delta=+0$, or $\epsilon = -1 + \epsilon'$, with $\epsilon'$ the impact of the propagation of the next carry. Recursively, we have $\epsilon \leq +1$ and $\Delta\leq+0$. We can thus always build the following two-classes distinguisher: \begin{equation}\label{eq:ravi}
        \text{For the bias }+\nint*{\frac{qk}{2^d}},\text{ with }k=2^j,j\in\mathbb{N},\{\Delta=+1\rightarrow m_j=0,\Delta\leq+0\rightarrow m_j=1\}.
    \end{equation}
\end{proof}

A consequence of Equation \ref{eq:ravi} is that the efficiency of the Ravi et al.\cite{ravi2021exploiting} targeted bit flip attack is 1 bit per query for any value of $d$. Equation \ref{eq:ravi} guarantees that there is no uncertainty about the value of the bit $m_j$. Hence, the targeted bit flip attack requires $d$ queries per instance of the function $\Deco(\alpha,d)$ in the attacked primitive in addition to the queries required to recover the original Hamming Weight of the message.


In a real-world scenario, the message is stored in memory and the queries can be parallelised. As a consequence, Ravi et al. \cite{ravi2021exploiting} attack can recover the entire message in $9$ traces if it is stored in 8-bit registers, $17$ traces if it is stored in 16-bit registers, only if $d$ is a power of 2 lesser or equal to the size of the register. Otherwise, the outputs of the function \Deco might overlap between two registers, and thus the attack cannot be completely parallelised. For example, this is the case in the second security level of FrodoKEM \cite{naehrig2017frodokem}, where $d=3$.

\subsection{The FrodoKEM example}
In their work, Ravi et al.\cite{ravi2021exploiting} detailed their attack for ML-KEM \cite{nist2023mlkem} and gave indications on how to extend it to SABER \cite{basso2020saber} and Frodo-640 \cite{naehrig2017frodokem} in their appendix. In this section, we focus on FrodoKEM \cite{naehrig2017frodokem}. We extend the work of Ravi et al.\cite{ravi2021exploiting} to cover the higher security parameters of FrodoKEM \cite{naehrig2017frodokem} as some peculiarities remain regarding how to attack them.

\begin{remark}
    As the modulus in FrodoKEM \cite{naehrig2017frodokem} is a power of two, we omit the rounding $\nint{\cdot}$ in the biases.
\end{remark}

\subsubsection{Existing targets}
As the attack is based on the Hamming Weight, the preferred target is microcontroller implementation, especially on the Cortex M4 family of microcontrollers. In the case of FrodoKEM \cite{naehrig2017frodokem}, we can cite as a reference the \emph{pqm4} \cite{pqm} Github repository, under the Tag ``Round 3''. However, this implementation only contains the first security level of FrodoKEM \cite{naehrig2017frodokem}, i.e., Frodo-640. The second level of security of FrodoKEM \cite{naehrig2017frodokem} was first implemented by Howe et al. \cite{Howe_Oder_Krausz_Guneysu_2018} with improved memory allocation. Further improvements were proposed by Bos et al. \cite{EPRINT:BFMOS18b}, using Single Instruction Multiple Data (SIMD) to parallelise some calculations. The first implementation with the third security level, that is, Frodo-1344, is found in a work by Bos et al. \cite{Bos_Bronchain_Custers_Renes_Verbakel_van_Vredendaal_2023} with a design centred on memory stack optimisation. Although this allows for the implementation of this third security level of FrodoKEM \cite{naehrig2017frodokem}, it comes with a cost in time performance.


Regarding the size of the registers used to store the output of the message decoding, the only publicly available implementation, \emph{pqm4} \cite{pqm}, uses 16-bit registers.

\subsubsection{How to perform the attack}
We follow the same procedure as the Ravi et al. \cite{ravi2021exploiting} targeted bit flip attack. In FrodoKEM \cite{naehrig2017frodokem}, the input of the function $\Deco$ in \Fdec at the beginning of \KDec is a matrix in $\mathbb{Z}_q^{\bar{m}\times\bar{n}}$. The decoding is performed on each coefficient of the matrix separately. To target a specific area of the message, we modify the ciphertext\footnote{With FrodoKEM \cite{naehrig2017frodokem} notation: $C$} $\textbf{v}$, which is itself a matrix in $\mathbb{Z}_q^{\bar{m}\times\bar{n}}$, by adding bias at the correct coefficient. We then compute \KDec with the biased ciphertext. The attack of Ravi et al. \cite{ravi2021exploiting} against FrodoKEM \cite{naehrig2017frodokem} is carried out as follows for each coefficient of \textbf{v}:\begin{enumerate}
    \item The Hamming Weight of the unbiased message is recovered with $\mathcal{O}_{SCA}$,
    \item The bias is injected within the ciphertext $\textbf{v}$ at the correct coefficient,
    \item The Hamming Weight of the biased message is recovered with a call to $\mathcal{O}_{SCA}$,
    \item Steps 2 and 3 are repeated until enough information is recovered to determine the value of the message bits linked to the targeted coefficient of $\textbf{v}$.
\end{enumerate}

\begin{remark}
    From an implementation point of view, the elements of $\mathbb{Z}_q$ are encoded in $16$ bits. This implies that the most significant bits of these coefficients are always $0$ since the modulus of Frodo-640 is lower than $2^{16}$. However, the ciphertexts are packed in bytes, where these extra zeros are removed. This implies that injecting bias into the ciphertext $\textbf{v}$ can impact up to $3$ bytes of the packed ciphertext. Unpacking the ciphertext does not alienate the value of the ciphertext coefficient, and thus the value of a biased ciphertext coefficient as well.

    For Frodo-976 and Frodo-1344, the packing and unpacking is simply a concatenation of coefficients of the matrix. The bias injection affects 2 bytes of the packed ciphertext.
\end{remark}

\subsubsection{Frodo-640} The first FrodoKEM security level uses the function \Deco with the parameter $d=2$, denoted $B$ in the FrodoKEM specification \cite{naehrig2017frodokem}. This implies mapping to two bits, and hence $2^d=4$ possible values. Table \ref{tab:frodo_bf_2} summarises the impact on the Hamming Weight of the output of \Deco when adding $\frac{q}{4}$, $\frac{q}{2}$ or $\frac{3q}{4}$ to the input compared to the Hamming Weight of the output without any input bias.

\begin{table}[H]
    \centering
    \caption{Impact of biasing the decoding function input with parameter $d=2$ on the Hamming Weight of the output}
    \label{tab:frodo_bf_2}
    \begin{tabular}{|c|c||c|c|c|}
        \hline
        Initial input & Initial mapping & $+\frac{q}{4}$ & $+\frac{q}{2}$ & $+\frac{3q}{4}$\\
        \hline
        $\llbracket -\frac{q}{8}, \frac{q}{8} \rrbracket$ & $00$ & $+1$ & $+1$ & $+2$ \\
        \hline
        $\llbracket \frac{q}{8}, \frac{3q}{8} \rrbracket$ & $01$ & $+0$ & $+1$ & $-1$ \\
        \hline
        $\llbracket \frac{3q}{8}, \frac{5q}{8} \rrbracket$ & $10$ & $+1$ & $-1$ & $+0$ \\
        \hline
        $\llbracket \frac{5q}{8}, \frac{7q}{8} \rrbracket$ & $11$ & $-2$ & $-1$ & $-1$ \\
        \hline
    \end{tabular}
\end{table}

We have the following distinguishers for each bias used in Table \ref{tab:frodo_bf_2}:\begin{itemize}
    \item $+\frac{q}{4}$: \begin{itemize}
        \item two-classes: $\{ +1 \rightarrow \centerdot0 ,(+0,-2)\rightarrow \centerdot1 \}$
        \item three-classes: $\{ +1 \rightarrow \centerdot0 , +0 \rightarrow 01 , -2 \rightarrow 11 \}$
    \end{itemize}
    \item $+\frac{q}{2}$: two-classes: $\{ +1 \rightarrow 0\centerdot , -1 \rightarrow 1\centerdot \}$
    \item $+\frac{3q}{4}$: \begin{itemize}
        \item two-classes: $\{ -1 \rightarrow \centerdot1 ,(+0,+2)\rightarrow \centerdot0 \}$
        \item three-classes: $\{ -1 \rightarrow \centerdot1 , +0 \rightarrow 10 , +2 \rightarrow 00 \}$
    \end{itemize}
\end{itemize}

Ravi et al. \cite{ravi2021exploiting} targeted bit flip attack can be performed in two queries with biases $+\frac{q}{2}$ and $+\frac{q}{4}$. Given the lists of two-classes distinguishers and Table \ref{tab:frodo_bf_2}, the two-classes distinguishers of $+\frac{q}{4}$ and $+\frac{q}{2}$ are complementary, as one gives away the Most Significant Bit (MSB) and the other the Least Significant Bit (LSB). However, we also notice that the bias $+\frac{3q}{4}$ can also be used alongside $+\frac{q}{2}$ or $+\frac{q}{4}$: \begin{itemize}
\item The two-classes distinguisher of the bias $+\frac{3q}{4}$ gives the LSB, and thus can be paired with the bias $+\frac{q}{2}$ which gives the MSB.
\item The three-classes distinguisher of the bias $+\frac{3q}{4}$ and the three-classes distinguisher of the bias $+\frac{q}{4}$ complement each other. The uncertainty in class $\{+1\rightarrow\centerdot0\}$ for bias $+\frac{q}{4}$ is solved by the two classes $\{+0\rightarrow10\}$ and $\{+2\rightarrow00\}$ for bias $+\frac{3q}{4}$. Similarly, the uncertainty in the class $\{-1\rightarrow\centerdot1\}$ linked to $+\frac{3q}{4}$ is solved by the two classes $\{+0\rightarrow01\}$ and $\{-2\rightarrow11\}$ linked to $+\frac{q}{4}$.
\end{itemize}


Frodo-640 can be breached by using ciphertext malleability with two Chosen Ciphertext queries and their associated oracle $\mathcal{O}_{SCA}$ queries per couple of message bits. The length of the message in Frodo-640 is $128$ bits. Full message recovery is possible in $128+Q$ queries to $\mathcal{O}_{SCA}$, $Q$ being the number of calls to $\mathcal{O}_{SCA}$ to recover the Hamming Weight of the entire unbiased message. In terms of traces, for an implementation using 16-bit registers and parallelising the attacks, the complete attack costs $17$ traces (only one trace is required to recover the Hamming Weight of the entire unbiased message). For an implementation using 8-bit registers, only $9$ traces are required.

\begin{remark}\label{rem:twoclassdistinguisher}
    An attacker can decide to recover only half the message bits using any two-classes distinguisher in Frodo-640 and brute-force the remaining $64$ bits of the message as all the two-classes distinguishers guarantee the recovery of at least one bit. This method allows for full message recovery with $64+Q$ queries to the oracle (or $5$ traces if 8-bit storage, $9$ if 16-bit storage) and a brute-force complexity of $2^{64}$ in the worst case, which is feasible with current computer technology.
\end{remark}

\subsubsection{Frodo-976} The second security parameters set of FrodoKEM uses a different compression factor for message encoding with $d=3$. Figure \ref{subfig:quad3} gives a sector representation of the \Deco function for this parameter. There are $2^d=8$ possible outputs encoded in $3$ bits. Table \ref{tab:frodo_bf_3} highlights the impact of adding $\frac{q}{8}$, $\frac{q}{4}$, $\frac{q}{2}$, or $\frac{5q}{8}$ to the input of \Deco on the Hamming Weight of its outputs.

\begin{table}[H]
    \centering
    \caption{Impact of biasing the decoding function input with parameter $d=3$ on the Hamming Weight of the output}
    \label{tab:frodo_bf_3}
    \begin{tabular}{|c|c||c|c|c|c|}
        \hline
        Initial input & Initial mapping & $+\frac{q}{8}$ & $+\frac{q}{4}$ & $+\frac{q}{2}$ & $+\frac{5q}{8}$\\
        \hline
        $\llbracket -\frac{q}{16}, \frac{q}{16} \rrbracket$ & $000$ & $+1$ & $+1$  & $+1$ & $+2$ \\
        \hline
        $\llbracket \frac{q}{16}, \frac{3q}{16} \rrbracket$ & $001$ &{$+0$} & $+1$ & $+1$ & {$+1$}\\
        \hline
        $\llbracket \frac{3q}{16}, \frac{5q}{16} \rrbracket$ & $010$ & $+1$ & {$+0$} & $+1$ & $+2$\\
        \hline
        $\llbracket \frac{5q}{16}, \frac{7q}{16} \rrbracket$ & $011$ & {$-1$} & {$+0$} & $+1$&{$-2$}\\
        \hline
        $\llbracket \frac{7q}{16}, \frac{9q}{16} \rrbracket$ & $100$ & $+1$ & $+1$ & {$-1$} & $+0$\\
        \hline
        $\llbracket \frac{9q}{16}, \frac{11q}{16} \rrbracket$ & $101$ & {$+0$} & $+1$ & {$-1$} & {$-1$}\\
        \hline
        $\llbracket \frac{11q}{16}, \frac{13q}{16} \rrbracket$ & $110$ & $+1$ & {$-2$} & {$-1$}& $+0$\\
        \hline
        $\llbracket \frac{13q}{16}, \frac{15q}{16} \rrbracket$ & $111$ & {$-3$} & {$-2$} & {$-1$} & {$-2$}\\
        \hline
    \end{tabular}
\end{table}

\medskip

\noindent We have the following distinguishers for each bias used in Table \ref{tab:frodo_bf_3}:
\begin{itemize}
    \item $+\frac{q}{8}$: \begin{itemize}
        \item two-classes: $\{ +1 \rightarrow \centerdot\centerdot0 ,(-3,-1,+0)\rightarrow \centerdot\centerdot1 \}$
        \item three-classes: $\{ +1 \rightarrow \centerdot\centerdot0 ,+0\rightarrow\centerdot01,(-3,-1)\rightarrow \centerdot11 \}$
        \item four-classes: $\{ +1 \rightarrow \centerdot\centerdot0 ,+0\rightarrow\centerdot01,-3\rightarrow111,-1\rightarrow011 \}$
    \end{itemize}
    \item $+\frac{q}{4}$: \begin{itemize}
        \item two-classes: $\{+1\rightarrow\centerdot0\centerdot,(+0,-2)\rightarrow\centerdot1\centerdot\}$
        \item three-classes: $\{+1\rightarrow\centerdot0\centerdot,+0\rightarrow01\centerdot,-2\rightarrow11\centerdot\}$ 
    \end{itemize}
    \item $+\frac{q}{2}$: two-classes: $\{-1\rightarrow1\centerdot\centerdot,+1\rightarrow0\centerdot\centerdot\}$
    \item $+\frac{5q}{8}$: \begin{itemize}
        \item two-classes: $\{(-2,-1,+1)\rightarrow\centerdot\centerdot1,(+0,+2)\rightarrow\centerdot\centerdot0\}$
        \item three-classes: $\{(-2,-1,+1)\rightarrow\centerdot\centerdot1,+2\rightarrow0\centerdot0,+0\rightarrow1\centerdot0\}$,$\{(-1,+1)\rightarrow\centerdot01,-2\rightarrow\centerdot11,(+0,+2)\rightarrow\centerdot\centerdot0\}$
        \item four-classes: $\{(-1,+1)\rightarrow\centerdot01,-2\rightarrow\centerdot11,+2\rightarrow0\centerdot0,+0\rightarrow1\centerdot0\}$,$\{-1\rightarrow101,+1\rightarrow001,-2\rightarrow\centerdot11,(-2,+0)\rightarrow\centerdot\centerdot0\}$
        \item five-classes: $\{-1\rightarrow101,+1\rightarrow001,+2\rightarrow0\centerdot0,-2\rightarrow\centerdot11,+0\rightarrow1\centerdot0\}$
    \end{itemize}
\end{itemize}


As in Frodo-640, we cannot perform a successful attack in only one $\mathcal{O}_{SCA}$ query per \Deco as no column in Table \ref{tab:frodo_bf_3} contains an eight-classes distinguisher. However, contrary to the attack on Frodo-640 where any couple of distinct biases is enough to recover the message bit values, it is not the case for Frodo-976. 

For example, when considering biases $+\frac{q}{8}$ and $+\frac{q}{4}$, they do not complement each other to solve the uncertainty of the class $\{+1\rightarrow\centerdot\centerdot0\}$ for the bias $+\frac{q}{8}$ or the uncertainty of the class $\{+1\rightarrow\centerdot0\centerdot\}$ for the bias $+\frac{q}{4}$. However, a query with bias $+\frac{q}{2}$ gives the missing information. In fact, using the biases $+\frac{q}{8},+\frac{q}{4}$ and $+\frac{q}{2}$ is equivalent to the Ravi et al. \cite{ravi2021exploiting} attack as generalised in Section \ref{subsec:gene}.


The targeted bit flip attack has a recovery efficiency of 1 bit per query. However, a better choice of biases allows for a more efficient bit recovery. We heuristically determined that querying $\mathcal{O}_{SCA}$ for $+\frac{q}{4}$ and $+\frac{5q}{8}$ allows us to build an eight-classes distinguisher and consequently recover every bit of the message with only two oracle queries per 3 bits of message, as shown in Table \ref{tab:eight_class}, where the uncertainties left by the first query with bias $+\frac{q}{4}$ are solved by the query with bias $+\frac{5q}{8}$.

\begin{table}[H]
    \centering
    \caption{Building a eight-classes distinguisher from $+\frac{q}{4}$ and $+\frac{5q}{8}$}
    \label{tab:eight_class}
    \begin{tabular}{|c|c|c|}
        \hline
        $+\frac{q}{4}$ & $+\frac{5q}{8}$ & possibilities \\
        \hline
        \multirow{4}{5em}{$+1\rightarrow\centerdot0\centerdot$} & $-1\rightarrow101$ & $101$ \\
        \cline{2-3}
         & $+1\rightarrow001$ & $001$ \\
         \cline{2-3}
         & $+2\rightarrow0\centerdot0$ & $000$ \\
         \cline{2-3}
         & $+0\rightarrow1\centerdot0$ & $100$ \\
         \hline
        \multirow{2}{5em}{$+0\rightarrow01\centerdot$} & $-2\rightarrow\centerdot11$ & $011$ \\
        \cline{2-3}
         & $+2\rightarrow0\centerdot0$ & $010$ \\
         \hline
        \multirow{2}{5em}{$-2\rightarrow11\centerdot$} & $-2\rightarrow\centerdot11$& $111$ \\
        \cline{2-3}
         & $+0\rightarrow1\centerdot0$ & $110$ \\
         \hline    
    \end{tabular}
\end{table}
%
%
A complete message recovery attack can be successfully performed against Frodo-976 with only $128+Q$ queries to the $\mathcal{O}_{SCA}$ oracle.


A peculiarity of Frodo-976 is that, since $d=3$ does not divide $8$ nor $16$, the output of some calls to the \Deco function will be split between two registers. For example, if 8-bit (resp. 16-bit) registers are used to store the message, in a 24-bit (resp. 48-bit) long word, that is, 3 registers, the first register is filled with the outputs of 3 (resp. 5) \Deco functions. The same applies to the last register. However, the middle register is filled with the outputs of 4 (resp. 6) \Deco functions that overlap with the first or last register. Although attacks on non-overlapping coefficients can be parallelised, the overlapping ones require specific investigation.


In the 8-bit register case, the overlap issue is represented in Figure \ref{fig:overlap8}, where the output of each \Deco is stored in big endian. \Deco 3 and \Deco 6 affect two registers, with one in common (Register 2). \Deco 1,4 and 7 can be attacked in parallel with two traces. The same applies to \Deco 2,5 and 8. The number of traces required for \Deco 3 and 6 can be reduced when specific biases are chosen. The biases $+\frac{q}{4}$ and $+\frac{3q}{4}$ do not affect the least significant bit of the output of \Deco. As a consequence, choosing one of those two biases leads to a non-overlapping impact even if the \Deco output is split between two registers. The number of traces required to attack the output of \Deco 3 and 6 in a parallel manner is therefore $3$. The total number of traces required to recover the message in an 8-bit implementation of Frodo-976 is $8$ with our optimised bias choices. 

\begin{remark}\label{rem:trace_ravi_8_976}
    The adaptation of the Ravi et al. \cite{ravi2021exploiting} attack is also affected by the overlap issue. However, the bias $+\frac{q}{2}$ only affects the MSB and $+\frac{q}{4}$ the middle bit and the MSB. This implies that attacks against \Deco 3 and \Deco 6 can be parallelised for the two higher bits of \Deco 3. Recovering the LSB of the output of \Deco 3 cannot be parallelised with the recovery of the last bit of \Deco 6. Thus, it requires an extra trace, for a total of 10.
\end{remark}

\begin{figure}[h]
    \centering
    \includegraphics[width=\textwidth]{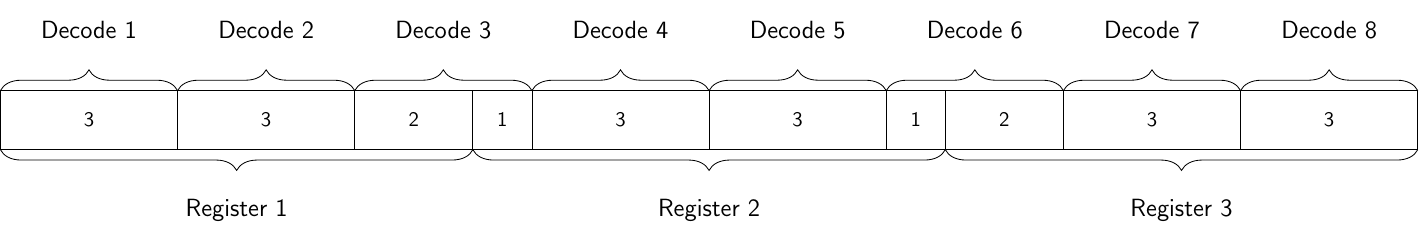}
    \caption{Overlap issue with 8 decodings and their 24-bit long output split on 3 registers}
    \label{fig:overlap8}
\end{figure}


In the 16-bit register case, the overlap issue is represented in Figure \ref{fig:overlap16}. \Deco 6 and \Deco 11 affect two registers and thus would require 4 traces. However, the middle register has only four complete \Deco compared to the five of the first and last registers. Using the observation that the biases $+\frac{q}{4}$ and $+\frac{3q}{4}$ do not affect the least significant bit of the output of \Deco, we reduce the number of required traces. We parallelise as follows: \Deco 1,7 and 12 (two traces), \Deco 2, 8 and 13 (two traces), \Deco 3, 9 and 14 (two traces), \Deco 4, 10 and 15 (two traces). We perform the two attacks on \Deco 5 and 16 in parallel with only one attack against \Deco 11 using the aforementioned biases (two traces). With the three remaining traces to sort the overlapping \Deco instances, and the original trace, the complexity of the parallelised attack against Frodo-976 implemented with 16-bit registers is 14 traces. 
\begin{remark}
    Remark \ref{rem:trace_ravi_8_976} also applies to the 16-bit register case for the adaptation of Ravi et al. \cite{ravi2021exploiting} attack. The parallelisation in this case requires not one but two extra traces, for a total of 19. 
\end{remark}

\begin{figure}[h]
    \centering
    \includegraphics[width=\textwidth]{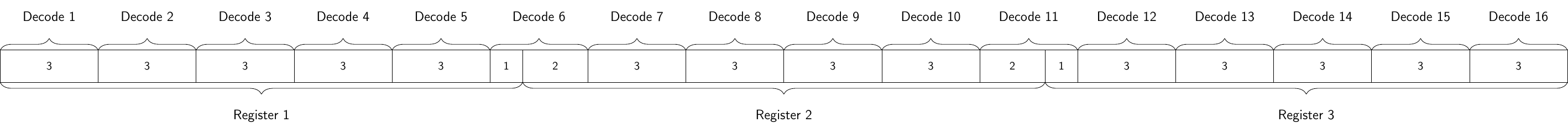}
    \caption{Overlap issue with 16 decodings and their 48-bit long output split on 3 registers}
    \label{fig:overlap16}
\end{figure}


\subsubsection{Frodo-1344} The highest security level of FrodoKEM uses $d=4$. With this parameter, the \Deco function maps to $2^d=16$ possible outputs encoded on $4$ bits. Table \ref{tab:frodo_bf_4} shows the impact of adding $+\frac{2^j}{16}$ with $j\in\llbracket 0,3 \rrbracket$ or $+\frac{5q}{8}$ to the input of \Deco.

\eject

\begin{table}[ht]
    \centering
    \caption{Impact of biasing the decoding function input with parameter $d=4$ on the Hamming Weight of the output}
    \label{tab:frodo_bf_4}
    \small
    \begin{tabular}{|c|c||c|c|c|c|c|}
        \hline
        Initial input & Initial mapping & $+\frac{q}{16}$ & $+\frac{q}{8}$ & $+\frac{q}{4}$ & $+\frac{q}{2}$ & $+\frac{5q}{8}$\\
        \hline
        $\llbracket -\frac{q}{32}, \frac{q}{32} \rrbracket$ & $0000$ & $+1$ & $+1$  & $+1$ &  $+1$ & $+2$\\
        \hline
        $\llbracket \frac{q}{32}, \frac{3q}{32} \rrbracket$ & $0001$ & $+0$ & $+1$ & $+1$ &  $+1$ & $+2$\\
        \hline
        $\llbracket \frac{3q}{32}, \frac{5q}{32} \rrbracket$ & $0010$ & $+1$ & $+0$ & $+1$ &  $+1$ & $+1$\\
        \hline
        $\llbracket \frac{5q}{32}, \frac{7q}{32} \rrbracket$ & $0011$ & $-1$ & $+0$ & $+1$ &  $+1$ & $+1$\\
        \hline
        $\llbracket \frac{7q}{32}, \frac{9q}{32} \rrbracket$ & $0100$ & $+1$ & $+1$  & $+0$ &  $+1$ & $ +2$\\
        \hline
        $\llbracket \frac{9q}{32}, \frac{11q}{32} \rrbracket$ & $0101$ & $+0$ & $+1$  & $+0$ & $+1$ & $+2$\\
        \hline
        $\llbracket \frac{11q}{32}, \frac{13q}{32} \rrbracket$ & $0110$ & $+1$ & $-1$ & $+0$ &  $+1$ & $-2$\\
        \hline
        $\llbracket \frac{13q}{32}, \frac{15q}{32} \rrbracket$ & $0111$ & $-2$ & $-1$  & $+0$ & $+1$& $-2$\\
        \hline
        $\llbracket \frac{15q}{32}, \frac{17q}{32} \rrbracket$ & $1000$ & $+1$ & $+1$ &  $+1$ & $-1$ & $+0$\\
        \hline
        $\llbracket \frac{17q}{32}, \frac{19q}{32} \rrbracket$ & $1001$ & $+0$ & $+1$ &  $+1$ &  $-1$& $ +0$\\
        \hline
        $\llbracket \frac{19q}{32}, \frac{21q}{32} \rrbracket$ & $1010$ & $+1$ & $+0$ &  $+1$ &  $-1$& $-1$\\
        \hline
        $\llbracket \frac{21q}{32}, \frac{23q}{32} \rrbracket$ & $1011$ & $-1$ & $+0$ &  $+1$ &  $-1$& $-1$\\
        \hline
        $\llbracket \frac{23q}{32}, \frac{25q}{32} \rrbracket$ & $1100$ & $+1$ & $+1$ &  $-2$ &  $-1$ & $+0$\\
        \hline
        $\llbracket \frac{25q}{32}, \frac{27q}{32} \rrbracket$ & $1101$ & $+0$ & $+1$ &  $-2$ &  $-1$ & $+0$\\
        \hline
        $\llbracket \frac{27q}{32}, \frac{29q}{32} \rrbracket$ & $1110$ & $+1$ & $-3$ &  $-2$ & $-1$ & $-2$\\
        \hline
        $\llbracket \frac{29q}{32}, \frac{31q}{32} \rrbracket$ & $1111$ & $-4$ & $-3$ &  $-2$ &  $-1$ & $-2$\\
        \hline
    \end{tabular}
\end{table}


The biases $+\frac{q}{16},\ +\frac{q}{8},\ +\frac{q}{4}$ and $+\frac{q}{2}$ are used to perform the Ravi et al. \cite{ravi2021exploiting} targeted bit flip attack. The efficiency of this attack remains at 1 bit per chosen ciphertext query. We heuristically determined that the minimum number of biases for a complete recovery of the output of $\Deco(\alpha,d)$ is 3. Combining the biases $+\frac{q}{16}$, $+\frac{q}{4}$ and $+\frac{5q}{8}$ allows us to build a sixteen-classes distinguisher and perform the attack successfully with only three Chosen Ciphertext queries to $\mathcal{O}_{SCA}$ per \Deco. Table \ref{tab:sixteen_class} describes the resolution of the uncertainties when the first bias is $+\frac{q}{4}$, the second is $+\frac{5q}{8}$ and the last\footnote{The order of the biases is for readability and has no impact on the attack} $+\frac{q}{16}$.

\begin{table}[ht]
    \centering
    \caption{Building a sixteen-classes distinguisher from $+\frac{q}{4}$, $+\frac{5q}{8}$ and $+\frac{q}{16}$}
    \label{tab:sixteen_class}
    \begin{tabular}{|c|c|c|c|}
        \hline
        $+\frac{q}{4}$ & $+\frac{5q}{8}$ & $+\frac{q}{16}$ & possibilities \\
        \hline
        \multirow{8}{5em}{$+1\rightarrow\centerdot0\centerdot\centerdot$} & \multirow{2}{5em}{$+2\rightarrow0\centerdot0\centerdot$}& $+1\rightarrow\centerdot\centerdot\centerdot0$ & $0000$ \\
        \cline{3-4}
        & & $+0\rightarrow\centerdot\centerdot01$ & $0001$ \\
         \cline{2-4}
        & \multirow{2}{5em}{$+1\rightarrow001\centerdot$} & $+1\rightarrow\centerdot\centerdot\centerdot0$ & $0010$ \\
         \cline{3-4}
         & & $-1\rightarrow\centerdot011$ & $0011$ \\
         \cline{2-4}
         & \multirow{2}{5em}{$+0\rightarrow1\centerdot0\centerdot$} & $+1\rightarrow\centerdot\centerdot\centerdot0$ & $1000$\\
         \cline{3-4}
         & & $+0\rightarrow\centerdot\centerdot01$ & $1001$\\
         \cline{2-4}
         & \multirow{2}{5em}{$-1\rightarrow101\centerdot$} & $+1\rightarrow\centerdot\centerdot\centerdot0$ & $1010$ \\
         \cline{3-4}
         & & $-1\rightarrow\centerdot011$ & $1011$\\
         \hline
        \multirow{4}{5em}{$+0\rightarrow01\centerdot\centerdot$} & \multirow{2}{5em}{$+2\rightarrow0\centerdot0\centerdot$}& $+1\rightarrow\centerdot\centerdot\centerdot0$ & $0100$ \\
        \cline{3-4}
         & &$+0\rightarrow\centerdot\centerdot01$ & $0101$ \\
         \cline{2-4}
         & \multirow{2}{5em}{$-2\rightarrow\centerdot11\centerdot$}& $+1\rightarrow\centerdot\centerdot\centerdot0$ & $0110$ \\
         \cline{3-4}
         & & $-2\rightarrow0111$ & $0111$ \\
         \hline
        \multirow{4}{5em}{$-2\rightarrow11\centerdot\centerdot$} &\multirow{2}{5em}{$+0\rightarrow1\centerdot0\centerdot$} & $+1\rightarrow\centerdot\centerdot\centerdot0$& $1100$ \\
        \cline{3-4}
        & & $+0\rightarrow\centerdot\centerdot01$ & $1101$ \\
        \cline{2-4}
        & \multirow{2}{5em}{$-2\rightarrow\centerdot11\centerdot$} & $+1\rightarrow\centerdot\centerdot\centerdot0$ & $1110$ \\
        \cline{3-4}
         & & $-4\rightarrow1111$ & $1111$ \\
         \hline    
    \end{tabular}
\end{table}


There are 95 possible choices of 3 biases to perform a complete \Deco output recovery. Frodo-1344 can be attacked for complete message recovery in only $192+Q$ Chosen Ciphertext queries with the help of the $\mathcal{O}_{SCA}$ oracle. For an implementation using 16-bit registers, thanks to the parallelisation of the attacks, only $13$ traces are required for a complete message recovery. For an implementation using 8-bit registers, only $7$ traces are required.

\subsection{Observations for higher values of $d$}
In this work, we determined heuristically the values of the optimal biases. However, we observed patterns in the choice of these biases. In this section, we discuss some of them and how they could affect message recovery for higher values of the parameter $d$. Other observations are conjectures; we leave a more thorough study for future work.

\begin{theorem}
    Any choice of biases for complete message recovery includes at least one bias $+\frac{qk}{2^d}$ with $k$ an odd number.
\end{theorem} 

\begin{proof}
    We define the binary decomposition of a bias $k$ as $k=\sum_{i=0}^{d-1}2^ik_i$. Similarly, the binary decomposition of the output of $\Deco(\alpha,d)$ is $m=\sum_{i=0}^{d-1}2^im_i$.
    \begin{equation} \label{eq:decal}
        \Deco(\alpha+\frac{qk}{2^d},d) \equiv \Deco(\alpha,d) + k \equiv \Deco(\alpha,d) + k_0 +2k_1 + \dots+2^{d-1}k_{d-1} \ mod\ 2^d.
    \end{equation}
    With the decomposition described in Equation \ref{eq:decal}, the calculation of the impact of the bias on the output of $\Deco$ can be performed by taking the impact of adding each $2^ik_i$ one after the other. Consequently, we can take each $k_i$ separately, compute a temporary $\Delta$ according to Equation \ref{eq:carryprop}, do the same with $k_{i+1}$ but taking into account the previous modification performed if $k_i=1$ to compute the next temporary $\Delta$. Given the representation in sectors, this is equivalent to performing a first rotation of $2^ik_i$ sectors, computing the temporary $\Delta$, and then performing the next rotation by $2^{i+1}k_{i+1}$ to add the results. 
    
    Let us assume that $k_0=0$ in Equation \ref{eq:decal}, i.e. $k$ is an even number. Then, the impact of the bias on the output bit $m_0$ is null as $\Deco(\alpha,d) = \Deco(\alpha,d)+0$. Regardless of the value of the next $k_i$ such as $i>0$, they do not affect the output bit $m_0$. According to Equation \ref{eq:carryprop}, the next $2^ik_i$ can affect only $m_1$ to $m_{d-1}$. As a consequence, there is always an uncertainty in the value of $m_0$ as it is never altered by the bias added to the input of $\Deco$, and therefore has never an impact on the variation of the Hamming Weight of $m$. To build a $2^d$-classes distinguisher, it is mandatory to have at least one bias $+\frac{qk}{2^d}$ with $k_0=1$.
\end{proof}

A comparison between the efficiency of the Ravi et al. \cite{ravi2021exploiting} targeted bit flip attack and a more optimal choice of biases is given in Figure \ref{fig:compar}. We heuristically determined the minimal number of biases required for a complete recovery by performing an exhaustive search for increasing values of the number of biases.

\begin{figure}[H]
    \centering
    \includegraphics[width=0.6\linewidth]{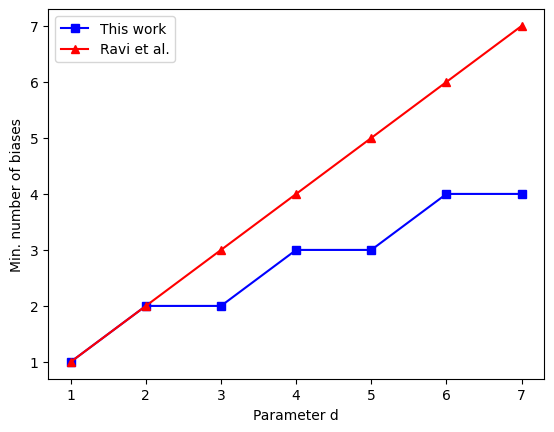}
    \caption{Comparison of the number of biases required for a complete output recovery between Ravi et al. \cite{ravi2021exploiting} and this work}
    \label{fig:compar}
\end{figure}



\section{Countermeasure}\label{sec:countermeasure}
\subsection{Overall strategy}
Ravi et al. \cite{ravi2021exploiting} claim that generic countermeasures such as shuffling or masking cannot protect against their attack. In this section, we propose reusing the attack as a defence mechanism. The rationale for the design is to apply the attack in a random but controlled manner. This will randomly flip the message bits to a point where an attacker cannot differentiate between his own attack and the countermeasure, thus counteracting the attack locally.


This countermeasure has an effect on the Hamming Weight of the message similar to the generic masking countermeasure. However, contrary to masking where each share is processed separately, in our case only one ``share'' is processed. This can reduce the cost of the implementation, as masking is known to have significant overhead due to its multiple shares to process. 

In this section, we present the countermeasure for the LPR framework and also detail its application to specific algorithms. 

\subsection{Application to the decryption}
To perform the countermeasure, we randomly choose a bias among a predefined set of biases and add it to the input of the \Deco function. This can be performed by sampling random bits and using a map between those bits and the set of biases.


We can reuse existing parts of the LPR framework to perform the two operations (sampling and mapping). Random sampling is already performed in LPR during the message sampling in the \KEnc procedure as described in Algorithm \ref{alg:fot}. For the mapping, it is performed by the \Enco function used in the \Fenc procedure. The message is then added to the received ciphertext $\textbf{v}$ in the same way as the message is added to the temporary ciphertext of $\textbf{v}$ during the \Fenc procedure.


In Algorithm \ref{alg:counterpke} a colour code is used to highlight the provenances of each part of the countermeasure. In \textcolor{bleu}{blue} are the functions reused from \KEnc, in \textcolor{vert}{green} are the ones reused from \Fenc.

\begin{algorithm}[H]
    \caption{Decryption with countermeasure}
    \label{alg:counterpke}
    \DontPrintSemicolon
    \Fn{\FdecCM{$ct$,$sk$}}{
        \textcolor{bleu}{$mask\leftarrow\mathcal{U}(\mathcal{B}^{32})$}\;
        \textcolor{vert}{$\textbf{vm} = \textbf{v} + \Enco(mask,d)$}\;
        $\textbf{xm}' = (\textcolor{vert}{\textbf{vm}} - \textbf{u}\times\textbf{s})$\;
        $mm= \Deco(\textbf{xm}',d)$\;
        \KwRet{$mm,mask$}\;
    }
\end{algorithm}

\medskip 

\noindent Due to ciphertext malleability, we have the following theorem:
\begin{theorem}
    Let $m$ be the message output of \Fdec{$ct$,$sk$}. Let $mm$ be the message output and $mask$ the mask output of \FdecCM{$ct$,$sk$}. Then, for each sub-message $m_i$, $mm_i$ and $mask_i$ of length $d$ bits, \begin{equation}mm_i=mask_i+ m_i\text{ mod }2^d.
    \end{equation} 
\end{theorem}
\begin{proof}
    This is a direct consequence of ciphertext malleability. Let us consider the \Enco and \Deco functions for the parameter $d$. We use the sector representation for \Deco. The function \Enco maps the bit words of value $k$ between $0$ and $2^d-1$ to the corresponding values $\frac{qk}{2^d}$. When used as biases on the input of \Deco, these values rotate the sectors, as discussed in Section \ref{subsec:gene}. As each sector covers $\frac{qk}{2^d}$ values, applying any of the aforementioned biases will rotate the sectors in a way that retains the set of the bounds of every sector. However, each sector now corresponds to its previous value $+k \text{ mod }2^d$. We then take $k=mask_i$ and reapplied this proof to the other sub-messages of $mask$.
\end{proof}

\begin{remark}
    If the parameter $d$ used for \Enco and \Deco is $1$, then $mm=m\oplus mask$ as $\forall(a,b)\in\{0,1\}, a+b\text{ mod }2 \Leftrightarrow a\oplus b $. This is the case for ML-KEM and SABER.
\end{remark}

\subsection{Extension to the remaining decapsulation}
Due to the Fujisaki-Okamoto Transform, the message is later re-encrypted during the \KDec procedure. In this section, we discuss how to extend our countermeasure to the remainder of the decapsulation.

\subsubsection{Seed generation}
The message is used to generate the random seed used in \Fenc. This is done by using a PRF on the message and a token derived from the public key. In both FrodoKEM \cite{naehrig2017frodokem}, ML-KEM \cite{nist2023mlkem}, and SABER \cite{basso2020saber}, this is performed with a call to a Keccak instance taking the message concatenated with a hash of the public key as input. This is central within the correctness of the FOT. If the message is altered, the seed generated during the decapsulation will be significantly different from the one used during the encapsulation, resulting in a different ciphertext altogether. Therefore, we need to ensure the correctness of the seed.


The application of our countermeasure results in a masked message. This masking is an arithmetic masking modulo $2^d$, or a boolean masking if $d=1$. Seed generation uses a hash function, implying computations with boolean logic. To perform it in a secure manner, if $d\neq1$ , we must first convert the masking from arithmetic to boolean (A2B conversion). Once done, this can be fed into a secure masked hash function. However, such an implementation of the hash function is costly but necessary.


\subsubsection{Re-encryption}
To further extend the countermeasure, we propose to perform re-encryption on the masked message but with the correct seed. The calculations of $\textbf{u}$ and $\textbf{v}'$ during \Fenc are therefore not altered by our countermeasure. However, we have $\textbf{v}''=\textbf{v}'+\Enco(mm)$, resulting in a different ciphertext.


After the re-encryption is performed, the generated ciphertext is compared with the received one. To ensure correctness, we have to ``correct'' $\textbf{v}''$ so it corresponds to $\textbf{v}$. This correction depends on the linearity of the \Enco function that is tied to $q$. It is therefore primitive-dependent. If $2^d$ divides $q$, then the rounding in Equation \ref{eq:decomp} is not necessary, since both $\beta$ and $\frac{q}{2^d}$ are integers. Multiplying an integer by a constant scalar is linear.


Both SABER \cite{basso2020saber} and FrodoKEM \cite{naehrig2017frodokem} have this property. In both cases, the modulo used within \Enco is a multiple of the compression factor, and consequently the rounding is not needed. This is not the case for ML-KEM \cite{nist2023mlkem}.

\vspace{-2ex}

\paragraph{Linear encoding}
 As the function \Enco is linear, to compute the correct ciphertext one can use the encoding of the mask on the generated ciphertext. 
 For FrodoKEM \cite{naehrig2017frodokem}, when performing the re-encryption during the decapsulation, instead of lines 11 and 12 in Algorithm \ref{alg:lpr}, we compute\footnote{With FrodoKEM \cite{naehrig2017frodokem} notation: $\textbf{C}' = \textbf{V} + \text{Frodo.Encode}(mm) - \text{Frodo.Encode}(mask).$} \begin{equation}
    \textbf{v} = \textbf{v'} + \Enco(mm,d) - \Enco(mask,d).
\end{equation}
For SABER \cite{basso2020saber}, the \Enco function is not defined but its effect is still used. An interesting detail is that the encoding is not added but subtracted. Thus, when performing the re-encryption during the decapsulation, instead of lines 11 and 12 in Algorithm \ref{alg:lpr}, we compute\footnote{With SABER \cite{basso2020saber} notation: $c_m = (v'+ h_1 - 2^{\epsilon_p-1}mm + 2^{\epsilon_p-1}mask\text{ mod }p) \gg (\epsilon_p-\epsilon_T).$} \begin{equation}
    \textbf{v} = \textbf{v'} - \Enco(mm,d) + \Enco(mask,d).
\end{equation}

\paragraph{Non-linear encoding}
ML-KEM \cite{nist2023mlkem} uses a Solinas/Crandall prime. Therefore, it cannot be divided by $2$ and as a result, the encoding is not linear. There are several options. We take inspiration from existing work on masked implementations of ML-KEM \cite{bos2021masking,cryptoeprint:2022/058}. 

The first idea is to ensure correctness by considering the encoding to be linear and then to further correct in specific cases. When performing the re-encryption during the decapsulation, instead of lines 11 and 12 in Algorithm \ref{alg:lpr} we compute\footnote{With ML-KEM \cite{nist2023mlkem} notation: $v' = \vec{t}^{\perp}\cdot \vec{r} + e_1 + Decompress_q(mm,1) + Decompress_q(mask,1).$}
\begin{equation}\label{eq:corr1}
    \textbf{v} = \textbf{v}' + \Enco(mm,1) + \Enco(mask,1). 
\end{equation}
However, as shown in Table \ref{tab:mlkembias}, there is a small bias when both $mask$ and $mm$ are equal to $1$. Table \ref{tab:mlkembias} uses the following legend:
\begin{itemize}
    \item \begin{tikzpicture}
        \node[draw=black,rectangle,fill=vert,text=white]{white};
    \end{tikzpicture} : The result is equal with $\Enco(m=mask\oplus mm,1)$
    \item \begin{tikzpicture}
        \node[draw=black,rectangle,fill=rouge,text=white]{white};
    \end{tikzpicture} : The result differs from $\Enco(m,1)$
\end{itemize}
\begin{table}[H]
    \caption{Comparing $\protect\Enco(mm,1)+\protect\Enco(mask,1)$ with $\protect\Enco(m,1)$ }
    \label{tab:mlkembias}
    \centering
    \large
    \begin{tabular}{|c|c|c|}
    \hline
         & $mm_{i}=0$ & $mm_{i}=1$ \\ 
    \hline
         $mask_i = 0$ & \cellcolor{vert}\textcolor{white}{\num{0}+\num{0}} & \cellcolor{vert}\textcolor{white}{\num{1665}+\num{0}}\\
    \hline
         $mask_i = 1$ & \cellcolor{vert}\textcolor{white}{\num{0}+\num{1665}} & \cellcolor{rouge}\textcolor{white}{\num{1665}+\num{1665}=\num{1}$\neq$\num{0}} \\
    \hline
    \end{tabular}
\end{table}
There are two ways to further correct this bias. The first relies on the masked implementation of Heinz et al. \cite{cryptoeprint:2022/058}. They propose to subtract the logic AND of the mask and the masked message from the computations of $\textbf{v}'$ before performing the ciphertext comparison. However, if not performed securely, this AND will reveal information on the message.


The second way is inspired by the masked implementation of Bos et al. \cite{bos2021masking}. Instead of comparing the ciphertexts by compressing the generated ciphertext and checking for a strict equality between $ct$ and $ct'$, they propose to compare the generated ciphertext directly with the decompression of the received one. However, since compression and decompression are lossy, they do not check for a strict equality but rather that the maximum distance between both ciphertexts does not exceed a specific threshold. To compensate for the bias, we suggest increasing the threshold for ciphertext $\textbf{v}$ by one.


We propose a new alternative to these methods. Instead of adding the decompression of the mask and then correct the resulting bias, we propose the following formula\footnote{With ML-KEM \cite{nist2023mlkem} notation: $v' = \vec{t}^{\perp}\cdot \vec{r} + e_1 + Decompress_q(mm,1) + (-1)^{mm}Decompress_q(mask,1). $}:\begin{equation}\label{eq:corr2}
    \textbf{v} = \textbf{v'} + \Enco(mm,1) + (-1)^{mm}\Enco(mask,1).
\end{equation}
This results in Table \ref{tab:mlkembiasnew}, with the same legend as Table \ref{tab:mlkembias}.
\begin{table}[H]
    \caption{Comparison between $\protect\Enco(mm,1) + (-1)^{mm}\protect\Enco(mask,1)$ and $\protect\Enco(m,1)$}
    \label{tab:mlkembiasnew}
    \centering
    \large
    \begin{tabular}{|c|c|c|}
    \hline
         & $mm_{i}=0$ & $mm_{i}=1$ \\ 
    \hline
         $mask_i = 0$ & \cellcolor{vert}\textcolor{white}{\num{0}+\num{0}} & \cellcolor{vert}\textcolor{white}{\num{1665}-\num{0}}\\
    \hline
         $mask_i = 1$ & \cellcolor{vert}\textcolor{white}{\num{0}+\num{1665}} & \cellcolor{vert}\textcolor{white}{\num{1665}-\num{1665}=\num{0}} \\
    \hline
    \end{tabular}
\end{table}
\subsection{Scalability}
Scalability depends on the linearity of the encoding used. For FrodoKEM \cite{naehrig2017frodokem} and SABER \cite{basso2020saber}, our countermeasure can be easily scaled to simulate higher orders of masking by simply generating new masks, encoding them, and adding them to the ciphertext. As a result, we end up with a higher order of masking for the message while still computing ``one share''.


For ML-KEM \cite{nist2023mlkem}, the scalability is less trivial as the encoding is non-linear. Adding $\nint*{\frac{q}{2}}$ several times in a row creates a growing bias as $2\nint*{\frac{q}{2}}\neq q$ due to the parity of $q$. The probability of having a failed bitflip resulting in an incorrect masking is low, but grows with this bias. One way of slowing the growth of this bias is to alternate between adding and subtracting $\nint*{\frac{q}{2}}$. This does not alleviate the growing complexity of the adaptations of Equations \ref{eq:corr1} and \ref{eq:corr2}. 

\subsection{Implementation recommendations}
\subsubsection{Mask generation} 
To generate $mask$, we use the same constraints as for the generic masking countermeasure. The share must be generated with a cryptographically secure RNG or a TRNG. In the \emph{pqm4} implementation \cite{pqm}, a hardware RNG is used to generate the message, we recommend using the same to generate $mask$.

\subsubsection{On-the-fly encoding}
A major drawback of post-quantum primitives is the size of the computations. Hence, we recommend implementing the function \Enco  (used at Line 3 of Algorithm \ref{alg:counterpke}) following an ``On-the-fly'' philosophy, i.e. encoding a coefficient of $mask$ and adding it to $\textbf{v}$ at the proper place before encoding the next coefficient. A similar approach should be preferred for other calls to the function \Enco, e.g. during the re-encapsulation for the encoding of $mm$. This design is inspired by the work of Bos et al. \cite{Bos_Bronchain_Custers_Renes_Verbakel_van_Vredendaal_2023}, where they use such an approach to optimise the memory usage of a FrodoKEM implementation on Cortex M4. As a consequence, we avoid storing encodings in memory, thus reducing memory stack consumption. This also benefits side-channel resistance, as attacking internal registers is much harder than attacking memory storage. The drawback is an increase in the complexity of implementation in time. In FrodoKEM \cite{naehrig2017frodokem}, ML-KEM \cite{nist2023mlkem} and SABER \cite{basso2020saber}, the function \Enco, when applied to the message and thus to $mm$ or $mask$, is a simple scalar multiplication.


\subsection{Vulnerability and compatibility with other countermeasures}
Although our countermeasure acts as a natural countermeasure to the use of ciphertext malleability \cite{ravi2021exploiting}, it is not a generic standalone countermeasure, as it does not affect the secret. Another issue resulting from the reuse of the primitive functions is that existing side-channel attacks against these functions, especially against the function \Enco \cite{10.1007/978-3-030-44223-1_11,9217595}, can circumvent the countermeasure. Thus, for a fully secured implementation, our method must be combined with other countermeasures.


As our countermeasure can be performed with minimal new code and, from a high-level perspective, reuses existing functions within LPR-based primitives, it is fully compatible with other secured implementations such as masking and/or shuffling. For a masked implementation, our countermeasure can be used to artificially increase the masking order of the message in the implementation.


We recommend combining our countermeasure with the generic shuffling countermeasure \cite{herbst2006aes,veyrat2012shuffling}, especially for calls to function \Deco and function \Enco. Shuffling is considered an effective countermeasure to protect the function \Enco \cite{ravi2024side,10.1007/978-3-030-44223-1_11}.


In a way, our countermeasure forces the attacker that was targeting the decoding function, for which generic countermeasures were proved ineffective, to target the encoding function where a generic countermeasure, shuffling, is known to be effective.

\begin{remark}
    Most of the state-of-the-art masked implementations of ML-KEM perform \Deco by first converting from arithmetic to boolean logic and then using boolean logic to compute the decoding. However, there are alternatives which compute decoding using arithmetic and then convert the result into boolean logic, such as \cite{cryptoeprint:2023/1220}. Our countermeasure can be used in this case to perform a cheaper arithmetic to boolean conversion by simply unmasking rather than converting. As we provide a first-order boolean masking with our countermeasure, unmasking in arithmetic does not expose the secret, but only one boolean share of it.
\end{remark}


\subsection{Discussion on the impact on other attacks}
Some attacks aim to recover the entire message using altered ciphertexts such as the one studied in Section \ref{sec:attack}. From a generic point of view, any attack that requires complete control over the ciphertext $v$ and using the same oracle $\mathcal{O}_{SCA}$ described in our paper is impacted.

\medskip

\noindent 
This applies to several attacks against ML-KEM categorised in \cite{ravi2024side}:
\begin{itemize}
    \item \emph{Binary Plaintext-Cheking} attacks \cite{ravi2020generic,shen2023find,qin2021systematic,buaetu2019misuse} use sparse ciphertexts to target single coefficients of the secret,
    \item \emph{Parallel Plaintext-Checking} attacks \cite{rajendran2023pushing,tanaka2023multiple} use sparse ciphertexts as well but target multiple coefficients of the secret at the same time.
\end{itemize}

\section{Conclusion}\label{sec:conclusion}
In this paper we recall an attack from Ravi et al. \cite{ravi2021exploiting} aimed at recovering the message in LPR-based cryptography primitives. We extended their work by detailing the attack on a post-quantum KEM, namely FrodoKEM, in Section \ref{sec:attack}. We also generalised their initial attack to any value of a parameter $d$ and discussed more optimal choices to reduce the number of traces required for a complete message recovery. We presented a novel countermeasure to this specific attack in Section \ref{sec:countermeasure}. The countermeasure is based on the same ciphertext malleability principle as the attack. We discussed the extension of the countermeasure to the entirety of the \KDec procedure, as well as its compatibility with other generic countermeasures and its scalability for several post-quantum primitives. As the countermeasure efficiency is local, we recommended its deployment in conjunction with a shuffling countermeasure.


The Ravi et al. \cite{ravi2021exploiting} attack principle is also used in several \emph{Deep Learning}-SCA (DL-SCA). It will be interesting to see how our countermeasure performs against a DL-SCA opponent. However, this will be implementation specific, as DL-SCA attacks tend to target specific \emph{determiner-leakages} due to the specific methods of computation used in some implementations. 

\paragraph{Acknowledgments}
The author thanks Yoan Rougeolle for his help with the experimental part of this work.

\noindent This work was realised thanks to the grant 2022156 from the Appel à projets 2022 thèses AID Cifre-Défense by the Agence de l'Innovation de Défense (AID), Ministère des Armées (French Ministry of Defence).


\bibliographystyle{fundam}
\bibliography{citations}
\end{document}